\documentclass{article}

\usepackage{arxiv}

\usepackage[utf8]{inputenc} 
\usepackage[T1]{fontenc}    
\usepackage{hyperref}       
\usepackage{url}            
\usepackage{booktabs}       
\usepackage{amsfonts}       

\usepackage{comment}
\usepackage{cite}
\usepackage{tabularx}

\usepackage[english]{babel}
\usepackage{amsmath}
\usepackage{amsthm}
\usepackage{amssymb}
\usepackage{algorithm}
\usepackage{algorithmic}
\usepackage{caption}
\usepackage{subcaption}
\usepackage{graphicx}
\graphicspath{ {./} }
\usepackage{xcolor}
\usepackage{tocloft}
\usepackage{thmtools,thm-restate}

\usepackage{listings}

\theoremstyle{plain}
\newtheorem{theorem}{Theorem}[section]

\newtheorem{invariant}[theorem]{Invariant}
\newtheorem{lemma}[theorem]{Lemma}

\newcommand{\R}{\mathbb{R}}
\newcommand{\A}{\mathcal{A}}

\title{Dynamic Maintenance of the Lower Envelope of Pseudo-Lines}

\author{
  Pankaj K. Agarwal \\
  Department of Computer Science,\\ 
  Duke University, Durham, NC 27708, USA\\ 
  pankaj@cs.duke.edu
  \And
  Ravid Cohen \\
  School of Computer Science,\\
  Tel-Aviv University, Tel-Aviv 69978, Israel
  \And
  Dan Halperin \\
  School of Computer Science,\\
  Tel-Aviv University, Tel-Aviv 69978, Israel
  \And
  Wolfgang Mulzer \\
  Institut f\"ur Informatik,\\ 
  Freie Universit\"at Berlin, D-14195 Berlin, Germany\\
  mulzer@inf.fu-berlin.de
}

\begin{document}
\maketitle

\begin{abstract}
 We present a fully dynamic data structure for the maintenance of lower envelopes of pseudo-lines. The structure has $O(\log^2 n)$ update time and $O(\log n)$ vertical ray shooting query time. To achieve this performance, we devise a new algorithm for finding the intersection between two lower envelopes of pseudo-lines in $O(\log n)$ time, using \emph{tentative} binary search; the  lower envelopes are special in that at $x=-\infty$ any pseudo-line contributing to the first envelope lies below every pseudo-line contributing to the second envelope. The structure requires $O(n)$ storage space.
\end{abstract}


\section{Introduction}

A set of pseudo-lines in the plane is a set of infinite $x$-monotone curves each pair of which intersects at exactly one point. Arrangements of pseudo-lines have been intensively studied in discrete and computational geometry;
see the recent survey on arrangements~\cite{hs-a-18} for a review of combinatorial bounds and algorithms for arrangements of pseudo-lines.
%
 {In this paper we consider the following problem:} Given $n$ pseudo-lines in the plane, dynamically maintain their lower envelope such that one can efficiently answer vertical ray shooting queries from $y=-\infty$. The dynamization is under insertions and deletions.  If we were given $n$ lines (rather than pseudo-lines) then we could have used any of several efficient data structures for the purpose~\cite{DBLP:conf/focs/BrodalJ02,DBLP:journals/jacm/Chan01,Overmars1981,BrodalJ00,KaplanTT01}; these are, however, not directly suitable for pseudo-lines. There are several structures that rely on shallow cuttings and can handle pseudo-lines~\cite{DBLP:journals/algorithmica/AgarwalM95,DBLP:journals/jacm/Chan10,DBLP:conf/soda/KaplanMRSS17}. The solution that we propose here is, however, considerably more efficient than what these structures offer.
We devise a fully dynamic  data structure with $O(\log^2 n)$ update-time, $O(\log n)$ vertical ray-shooting query-time, and $O(n)$ space for the maintenance of $n$ pseudo-lines.
The structure is a rather involved adaptation of the Overmars-van Leeuwen structure~\cite{Overmars1981} to our setting, which matches the performance of the original algorithm for the case of lines. 
The key innovation is a new algorithm for finding the
intersection between two lower envelopes of planar pseudo-lines
in $O(\log n)$ time, using \emph{tentative} binary search (where each pseudo-line in one envelope is ``smaller'' than every pseudo-line in the other envelope in a sense to be made precise below).
To the best of our knowledge this is the most efficient data structure for the case of pseudo-lines to date.





\section{Preliminaries}

Let $E$ be a finite family of pseudo-lines in the plane, and
let $\ell$ be a vertical line strictly to the left of 
the left-most intersection point between lines in $E$ (namely to the left of all the vertices of the arrangement $\A(E)$). The line $\ell$
defines a total order $\leq$ on the pseudo-lines in $E$, 
namely for $e_1, e_2 \in E$, we have $e_1 \leq e_2$ if and 
only if $e_1$ intersects $\ell$ below $e_2$. Since each 
pair of pseudo-lines in $E$ crosses exactly once, it follows 
that if we consider a vertical line $\ell'$ strictly to the 
right of the right-most vertex of $\A(E)$, the order 
of 
the intersection points between $\ell'$ and $E$, from 
bottom to top, is exactly reversed.

The \emph{lower envelope} $\mathcal{L}(E)$ of $E$ is the 
$x$-monotone curve obtained by taking the pointwise 
minimum of the pseudo-lines in $E$. Combinatorially, the 
lower envelope $\mathcal{L}(E)$ is a sequence of connected 
segments of the pseudo-lines in $E$, where the first
and last segment are unbounded. Two properties
are crucial for our data structure: (A) every pseudo-line 
contributes at most one segment to $\mathcal{L}(E)$; and  (B) the
order of these segments corresponds exactly to the 
order $\leq$ on $E$ defined above. 
In fact, our data structure works for every set of planar curves
with properties (A) and (B) (with an appropriate
order $\leq)$, even if they 
are not pseudo-lines in the strict sense.

We assume a computational model in which primitive
operations on pseudo-lines, such as computing the 
intersection point of two pseudo-lines or determining 
the intersection point of a pseudo-line 
with a vertical line 
can be performed in constant time.

\section{Data structure and operations}

\subparagraph{The tree structure.}
Our primary data structure is a balanced binary 
search tree $\Xi$.
Such a tree data structure supports insert and delete, each
in $O(\log n)$ time. 
The leaves of $\Xi$ contain 
the pseudo-lines, from left to right in the sorted order defined above. 
An internal node $v \in \Xi$ represents the 
lower envelope of the pseudo-lines in its subtree.  
More precisely, every leaf $v$ of $\Xi$ stores a 
single pseudo-line $e_v \in E$. For an inner node $v$ 
of $\Xi$, we write $E(v)$ for the set of 
pseudo-lines in the subtree rooted at $v$.
We denote the lower envelope of $E(v)$ by  
$\mathcal{L}\big(v\big)$.
The inner node $v$ has the following variables:
\begin{itemize}
\item $f$, $\ell$, $r$: a pointer to the parent, 
left child and right child of $v$, respectively;
\item $\max$: the maximum pseudo-line in $E(v)$;
\item $\Lambda$: 
a balanced binary search tree that stores the 
prefix or suffix of $\mathcal{L}(v)$ that 
is not on the lower envelope 
$\mathcal{L}(f)$ of the parent (in the root,
we store the lower envelope of $E$). 
The leaves of $\Lambda$ store the pseudo-lines that support the
segments on the lower envelope, 
with the
endpoints of the segments, sorted from left to right.
An inner node of $\Lambda$ stores the common point of the last segment
in the left subtree and the first segment in the right subtree.
We will need split and join operations on the binary trees, which can be implemented in $O(\log n)$ time.
\end{itemize}

\subparagraph*{Queries.} 
We now describe the query operations available 
on our data structure.
In a \emph{vertical ray-shooting query}, we are 
given a value $x_0 \in \R$, 
and we would like to find the pseudo-line 
$e \in E$ where the vertical
line $\ell: x = x_0$ intersects $\mathcal{L}(E)$.
Since the root of $\Xi$ explicitly stores
$\mathcal{L}(E)$ in a balanced binary search tree,
this query can be answered easily
in $O(\log n)$ time.

\begin{restatable}{restatable_lemma}{verticalRs}
\label{lem:vertical_rs}
Let $\ell: x = x_0$ be
a vertical ray shooting query. 
We can find the pseudo-line(s) where
$\ell$ intersects $\mathcal{L}(E)$ in $O(\log n)$
time.
\end{restatable}

\begin{proof}
Let $r$ be the root of $\Xi$. We perform an
explicit search for $x_0$ in 
$r.\Lambda$ and return the result. Since $r.\Lambda$
is a balanced binary search tree, this takes
$O(\log n)$ time.
\end{proof}

\subparagraph*{Update.} 
To insert or delete a pseudo-line
$e$ in $\Xi$, we follow
the method of Overmars and van 
Leeuwen~\cite{Overmars1981}. 
We delete or insert a leaf for $e$ in $\Xi$
using standard binary search tree techniques (the $v.\max$ pointers
guide the search in $\Xi$). As we go down,
we construct the lower envelopes for
the nodes hanging off the search path, using 
split and join operations on the $v.\Lambda$ trees. Going
back up, we recompute the information $v.\Lambda$ and
$v.\max$.
To update the $v.\Lambda$ trees, we need the following 
operation: given two lower envelopes $\mathcal{L}_\ell$ 
and $\mathcal{L}_r$, such that all pseudo-lines in $\mathcal{L}_\ell$
are smaller than all pseudo-lines in $\mathcal{L}_r$, 
compute the intersection point $q$ of $\mathcal{L}_\ell$
and $\mathcal{L}_r$. In the next section, we will see
how to do this in $O(\log n)$ time,
where $n$ is the size of $E$.
Since there are $O(\log n)$ nodes in $\Xi$ 
affected by an update, this procedure takes
$O(\log^2 n)$ time. More details can be found
in the literature~\cite{Overmars1981,PreparataSh85}.

\begin{lemma}
It takes $O(\log^2 n)$ to
insert or remove a pseudo-line in $\Xi$.
\end{lemma}

\section{Finding the intersection point of two lower envelopes}

Given two lower envelopes $\mathcal{L}_\ell$ 
and $\mathcal{L}_r$ such that all pseudo-lines 
in $\mathcal{L}_\ell$ are smaller than all 
pseudo-lines in $\mathcal{L}_r$, we would like to 
find the 
intersection point $q$ between $\mathcal{L}_\ell$ 
and $\mathcal{L}_r$. We assume that $\mathcal{L}_\ell$ 
and $\mathcal{L}_r$ are represented as balanced 
binary search trees. The leaves of $\mathcal{L}_\ell$ 
and $\mathcal{L}_r$ store the pseudo-line segments 
on the lower envelopes, sorted from left to
right. 
We assume that the pseudo-line segments in the
leaves are half-open, containing 
their right, but not their left endpoint in 
$\mathcal{L}_\ell$; and their left, but not their 
right endpoint in $\mathcal{L}_r$.\footnote{We 
actually store both endpoints in the trees,
but the intersection algorithm uses only one
of them, depending on the role the tree plays in the
algorithm.}
Thus, it is
uniquely determined which leaves of $\mathcal{L}_\ell$ 
and $\mathcal{L}_r$ contain the intersection point $q$.
A leaf $v$ stores the pseudo-line $\mathcal{L}(v)$
that supports the segment for $v$, 
as well as an endpoint $v.p$ of the segment,
namely the left endpoint
if $v$ is a leaf of  $\mathcal{L}_\ell$,
and the right endpoint if
$v$ is a leaf of $\mathcal{L}_r$.\footnote{If the
segment is unbounded, the endpoint
might not exist. In this case, we use
a symbolic endpoint at infinity that
lies below every other pseudo-line.}
An inner node $v$ 
stores the intersection point $v.p$ between the 
largest pseudo-line in the left subtree 
$v.\ell$ of $v$ and the smallest pseudo-line 
in the right subtree $v.r$ of $v$, together with 
the lower envelope $\mathcal{L}(v)$ of these two
pseudo-lines. These trees can be obtained by
appropriate split and join operations from the
$\Lambda$ trees stored in $\Xi$.

\begin{figure}[ht]
	\centering
    \includegraphics[width=\textwidth]
    {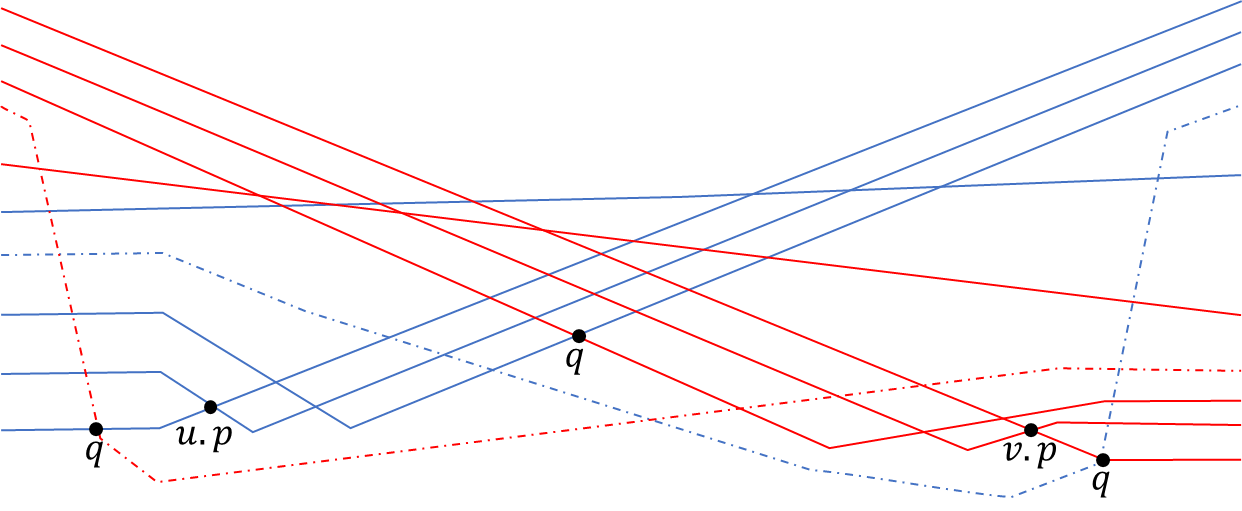}
    \caption{An example of Case~3. 
    $\mathcal{L}_\ell$ is blue; $\mathcal{L}_r$ is red.
    The solid pseudo-lines are fixed. 
    The dashed pseudo-lines are optional, namely, either
    none of the dashed pseudo-lines 
    exists or exactly one of them exists. $u.p$
    and $v.p$ are the current points;
    and Case~3 applies.
    Irrespective of the local
    situation at $u$ and $v$, the intersection point
    can be to the left of $u.p$, between $u.p$ and 
    $v.p$ or to the right of $v.p$, depending 
    on which one of the dashed pseudo-lines
    exists.} 
	\label{f:hard_case_pseudo_lines} 
\end{figure}

Let $u^* \in \mathcal{L}_\ell$ and 
$v^* \in \mathcal{L}_r$ be the leaves whose segments
contain $q$.  Let 
$\pi_\ell$ be the path in $\mathcal{L}_\ell$
from the root to $u^*$ and 
$\pi_r$ 
the path in $\mathcal{L}_r$ from the root to $v^*$. 
Our strategy is as follows: we simultaneously 
descend in $\mathcal{L}_\ell$ and in $\mathcal{L}_r$. 
Let $u$  be the current node in $\mathcal{L}_\ell$ 
and $v$ the current node in $\mathcal{L}_r$.
In each step, we perform a local test on 
$u$ and $v$ to decide how to proceed. 
There are three possible outcomes:
\begin{enumerate}
\item $u.p$ is on or above $\mathcal{L}(v)$: the 
intersection point $q$
is equal to or to the left of $u.p$. If $u$ is
an inner node, then $u^*$ cannot lie in $u.r$; 
if $u$ is a leaf, then $u^*$ lies strictly
to the left of $u$;
\item $v.p$ lies on or above $\mathcal{L}(u)$: 
the intersection point $q$
is equal to or to the right of $v.p$. 
If $v$ is an inner node, then $v^*$ cannot lie in
$v.\ell$; if $v$ is a leaf, then $v^*$ lies 
strictly to the right of $v$; 
\item $u.p$ lies below $\mathcal{L}(v)$ and $v.p$ lies 
below $\mathcal{L}(u)$: then, $u.p$ lies
strictly to the left of $v.p$ (since 
we are dealing with pseudo-lines). It must be
the case that $u.p$ is strictly to the left of
$q$ or $v.p$ is strictly to the right of $q$ (or both).
In the former case, if $u$ is an inner node,
$u^*$ lies in or to the right of $u.r$ and if $u$ is
a leaf, then $u^*$
is $u$ or a leaf to the right of $u$. In the
latter case, if $v$ is an inner node, $v^*$ lies in
or to the left of 
$v.\ell$ and if $v$ is a leaf, then $v^*$ is $v$
or a leaf to the left of $v$; see Figure~\ref{f:hard_case_pseudo_lines}.
\end{enumerate}
Overmars and van Leeuwen~\cite{Overmars1981,PreparataSh85}
describe a method for the case that $\mathcal{L}_\ell$ and
$\mathcal{L}_r$ contain lines. Unfortunately, it is not 
clear how their strategy applies in the more general
setting of pseudo-lines. The reason for this lies
in Case~3: in this case, it is not immediately
obvious how to proceed, because the correct step
might be either to go to $u.r$ or to $v.\ell$. 
In the case of lines, Overmars and van Leeuwen
can solve this ambiguity by comparing the slopes
of the relevant lines. For pseudo-lines, however,
this does not seem to be possible. For an example,
refer to Figure~\ref{f:hard_case_pseudo_lines}, where
the local situation at $u$ and $v$ does not determine
the position of the intersection point $q$. Therefore,
we present an alternative strategy.

\begin{figure}[ht]
    \centering
    \includegraphics[scale=0.85]{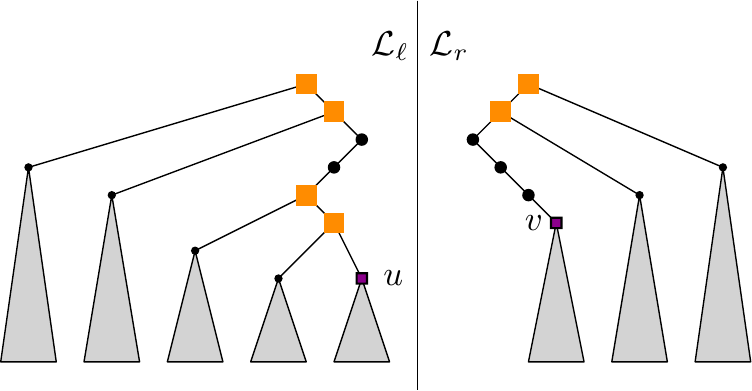}
    \caption{The invariant:
    the current search nodes are $u$ and $v$.
    \texttt{uStack} contains all nodes on the
    path from the root to $u$ where the path goes to a right
    child (orange squares), \texttt{vStack} contains all
    nodes from the root to $v$ where the path goes to a left child (orange squares). The final leaves $u^*$ and $v^*$ are in one of the
    gray subtrees; and at least one of them is under $u$ or under $v$.}
    \label{fig:invariant}
\end{figure}
\begin{figure}[ht]
    \centering
    \includegraphics[scale=0.67]{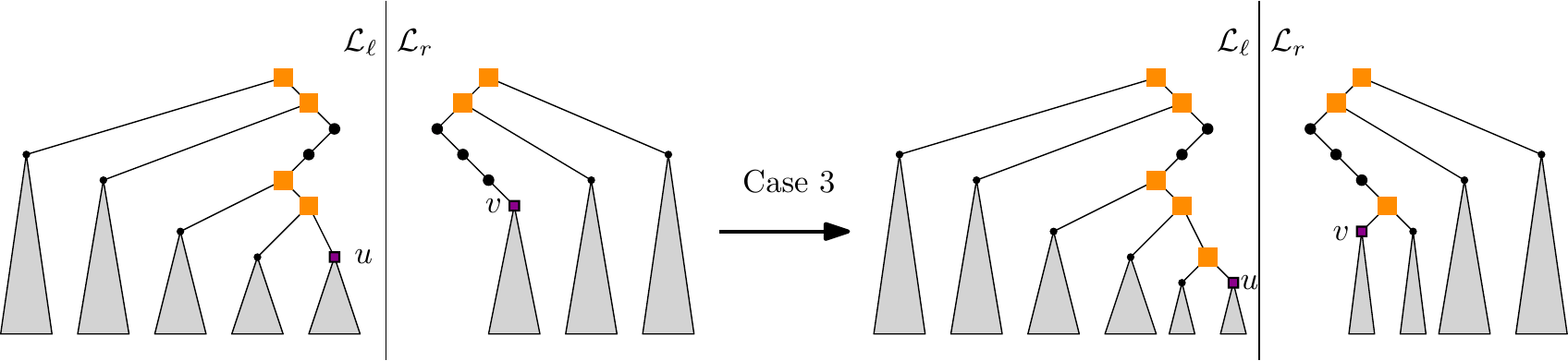}
    \caption{Comparing $u$ to $v$: in Case~3,
    we know that $u^*$ is in $u.r$ or $v^*$ is in $v.\ell$; we go to
    $u.r$ and to $v.\ell$.}
    \label{fig:case3}
\end{figure}
\begin{figure}[ht]
    \centering
    \includegraphics[scale=0.67]{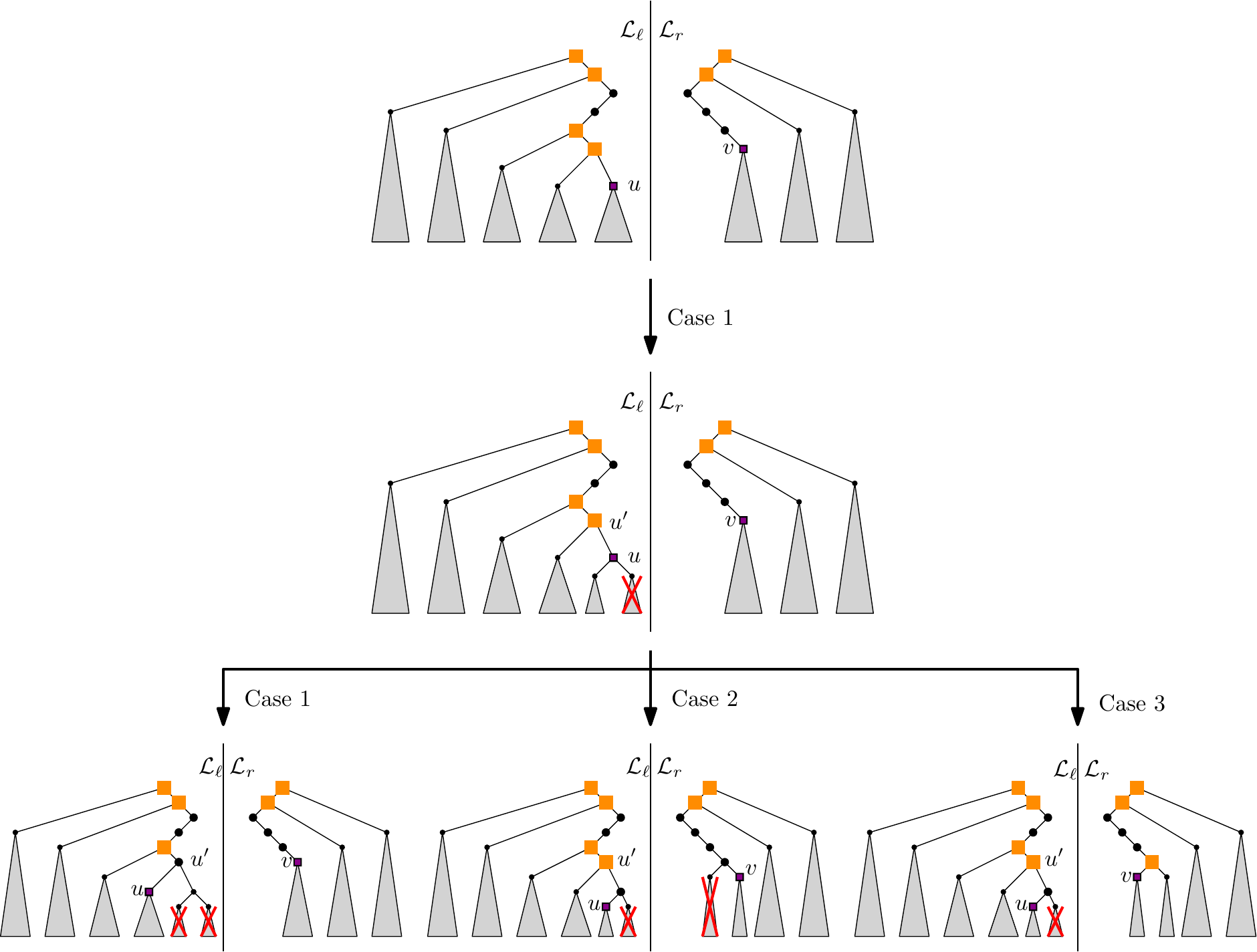}
    \caption{Comparing $u$ to $v$:
    in Case~1, we know that
    $u^*$ cannot be in $u.r$. 
    We compare $u'$ and $v$ to decide
    how to proceed: 
    in Case~1, we know that $u^*$ cannot be in $u'.r$; we go to
    $u'.\ell$; in Case 2, we know that $u^*$ cannot be in $u.r$ and that 
    $v^*$ cannot be in $v.\ell$; we go to $u.\ell$ and to $v.r$; in Case~3,
    we know that $u^*$ is in $u'.r$ (and hence in $u.\ell$) or in $v.\ell$;
    we go to $u.\ell$ and to $v.\ell$. Case~2 is not
    shown as it is symmetric.}
    \label{fig:case1}
\end{figure}

We will maintain 
the invariant that the subtree at $u$ contains $u^*$ or the
subtree at $v$ contains $v^*$ (or both). When comparing 
$u$ and $v$, one of the three 
cases occurs. In Case~3, $u^*$ must be
in $u.r$, or $v^*$ must be in $v.\ell$;
see Figure~\ref{fig:case3}. 
We move $u$ to $u.r$
and $v$ to $v.\ell$. One of these moves must be correct,
but the other move might be mistaken: we might have gone
to $u.r$ even though $u^*$ is in $u.\ell$ or to $v.\ell$
even though $v^*$ is in $v.r$. To correct this,
we remember the current $u$ in a stack \texttt{uStack} and
the current $v$ in a stack \texttt{vStack}, so that we can
revisit $u.\ell$ or $v.r$ if it becomes necessary. This leads
to the general situation shown in Figure~\ref{fig:invariant}:
$u^*$ is below $u$ or in a left subtree of a node
on $\texttt{uStack}$, and $v^*$ is below $v$ or in a right 
subtree of a node on $\texttt{vStack}$, and at least one of
$u^*$ or $v^*$ must be below $u$ or $v$, respectively. Now, if
Case~1 occurs when comparing $u$ to $v$, we can exclude the
possibility that $u^*$ is in $u.r$. Thus, $u^*$ might be in
$u.\ell$, or in the left subtree of a node in \texttt{uStack}; 
see Figure~\ref{fig:case1}.
To make progress, we now compare $u'$, the top of \texttt{uStack},
with $v$. Again, one of the three cases occurs. In Case~1,
we can deduce that going to $u'.r$ was mistaken, and we move
$u$ to $u'.\ell$, while $v$ does not move. In the other cases,
we cannot rule out that $u^*$ is to the right of $u'$, and we 
move $u$ to $u.\ell$, keeping the invariant that $u^*$ is either
below $u$ or in the left subtree of a node on \texttt{uStack}.
However, to ensure that the search progresses, we now must also
move $v$. In Case~2, we can rule out $v.\ell$, and we move 
$v$ to $v.r$. In Case~3, we move $v$ to $v.\ell$. In this way,
we keep the invariant and always make progress: in each step,
we either discover a new node on the correct search
paths, or we pop one erroneous move from one of the two stacks.
Since the total length of the correct search paths is
$O(\log n)$, and since we push an element onto the stack 
only when discovering a new correct node, 
the total search time is $O(\log n)$; see Figure~\ref{fig:algodemo} for an example run.
The following pseudo-code gives the details
of our algorithm, including all corner cases.

\begin{lstlisting}[backgroundcolor = \color{lightgray}, mathescape=true]
oneStep($u$, $v$)
     do compare($u$, $v$):
         Case 3: 
             if $u$ is not a leaf then 
                 uStack.push($u$); $u \leftarrow u.r$
             if $v$ is not a leaf then
                 vStack.push($v$); $v \leftarrow v.\ell$
             if $u$ and $v$ are leaves then
                 return $u = u^*$ and $v = v^*$
         Case 1:
             if uStack is empty then
                 $u \leftarrow u.\ell$
             else if $u$ is a leaf then
                 $u \leftarrow \texttt{uStack.pop().}\ell$
             else
                 $u' \leftarrow \texttt{uStack.top()}$
                 do compare($u'$, $v$)
                     Case 1: 
                         uStack.pop(); $u \leftarrow u'.\ell$
                     Case 2:
                         $u \leftarrow u.\ell$ 
                         if $v$ is not a leaf then
                             $v. \leftarrow v.r$
                     Case 3:
                         $u \leftarrow u.\ell$
                         if $v$ is not a leaf then
                             vStack.push($v$); $v \leftarrow v.\ell$
         Case 2:
            symmetric
\end{lstlisting}

We will show that the search procedure maintains
the following invariant:

\begin{invariant}
\label{inv:intersection}
The leaves in all subtrees $u'.\ell$, for 
$u' \in \texttt{uStack}$, together with the
leaves under $u$ constitute a contiguous 
prefix of the leaves in  $\mathcal{L}_\ell$, which
contains $u^*$. 
Also, the leaves in all subtrees
$v'.r$, $v' \in \texttt{vStack}$,  
together with the leaves under $v$ constitue a 
contiguous suffix of the leaves of $\mathcal{L}_r$,
which contains $v^*$. Furthermore,
either $u \in \pi_\ell$ or $v \in \pi_r$ (or both).
\end{invariant}

Invariant~\ref{inv:intersection} holds at the 
beginning, when both stacks are empty,
$u$ is the root of $\mathcal{L}_\ell$ and $v$ is the
root of $\mathcal{L}_r$. To show that the invariant
is maintained, we first consider the special case
when one of the two searches has already discovered
the correct leaf:
\begin{lemma}
\label{lem:leaf_inv}
Suppose that Invariant~\ref{inv:intersection} holds and
that Case~3 occurs when comparing $u$ to $v$.
If $u = u^*$, then
$v \in \pi_r$ and, if $v$ is not a leaf, $v.\ell \in \pi_r$.
Similarly, if $v = v^*$, then $u \in \pi_\ell$ and,
if $u$ is not a leaf, $u.r \in \pi_\ell$.
\end{lemma}

\begin{proof}
We consider the case $u = u^*$; the other case is symmetric. 
Let $e_u$ be the segment of $\mathcal{L}_\ell$ stored in $u$.
By Case~3, $u.p$ is strictly to the left of $v.p$. 
Furthermore, since $u = u^*$, the intersection point
$q$ must be on $e_u$. Thus, $q$ cannot
be on the right of $v.p$, because otherwise
$v.p$ would be a point on $\mathcal{L}_r$ that lies below $e_u$
and to the left of $q$, which is impossible.
Since $q$ is strictly to the left of $v.p$;
by Invariant~\ref{inv:intersection}, if $v$ is
an inner node, $v^*$ must be in $v.\ell$, and hence
both $v$ and $v.\ell$ lie on $\pi_r$. If $v$ is a
leaf, then $v = v^*$.
\end{proof}

We can now show that the invariant is maintained.

\begin{lemma}
The procedure \texttt{oneStep} either correctly
reports that $u^*$ and $v^*$ have been found, or it maintains Invariant~\ref{inv:intersection}. In the latter case, it either
pops an element from one of the two stacks, or it
discovers a new node on $\pi_\ell$ or $\pi_r$.
\end{lemma}

\begin{proof}
First, suppose Case~3 occurs. The invariant
that \texttt{uStack} and $u$ cover a prefix of
$\mathcal{L}_\ell$ and that \texttt{vStack} and $v$
cover a suffix of $\mathcal{L}_r$ is maintained.
Furthermore, if both $u$ and $v$ are inner nodes,
Case~3 ensures that $u^*$ is in $u.r$ or
to the right of $u$, or that $v^*$ is in $v.\ell$ or to the left of $v$. Suppose the former case
holds. Then, Invariant~\ref{inv:intersection} 
implies that $u^*$ must be in $u.r$, and
hence $u$ and $u.r$ lie on $\pi_\ell$. 
Similarly, in the second case, 
Invariant~\ref{inv:intersection} gives that
$v$ and $v.\ell$ lie in $\pi_r$.
Thus,
Invariant~\ref{inv:intersection} is maintained
and we discover a new node on $\pi_\ell$ or
on $\pi_r$.
%
%
Next, assume $u$ is a leaf and $v$ is an inner node.
If $u \neq u^*$, then as above, 
Invariant~\ref{inv:intersection} and Case~3 imply
that $v \in \pi_r$ and $v.\ell \in \pi_r$, 
and the lemma holds.
If $u = u^*$, the lemma follows from Lemma~\ref{lem:leaf_inv}.
The case that $u$ is an inner node and $v$ a 
leaf
is symmetric. If both $u$ and $v$ are leaves, Lemma~\ref{lem:leaf_inv}
implies that \texttt{oneStep} correctly  reports $u^*$ and
$v^*$.

Second, suppose Case~1 occurs. Then, 
$u^*$ cannot be in $u.r$, if $u$ is an
inner node, or  $u^*$ must be to the left
for a segment left of $u$, if $u$ is 
a leaf.
Now, if \texttt{uStack}
is empty, Invariant~\ref{inv:intersection}
and Case~1 imply that $u$ cannot be a leaf 
(because $u^*$ must be in the subtree of $u$)
and that $u.\ell$ is a new node on $\pi_\ell$.
Thus, the lemma holds in this case. 
Next, if $u$ is a leaf,  
Invariant~\ref{inv:intersection} and
Case~1 imply that $v \in \pi_r$. Thus, we pop
\texttt{uStack} and maintain the invariant; 
the lemma holds.
Now, assume that \texttt{uStack} is not
empty and that $u$ is not a leaf. 
Let $u'$ be the top of $\texttt{uStack}$.
First, if the comparison between $u'$ and $v$ results
in Case~1, then $u^*$ cannot be in
$u'.r$, and in particular, $u \not\in \pi_\ell$.
Invariant~\ref{inv:intersection} shows
that $v \in \pi_r$, 
and we pop an element from \texttt{uStack},
so the lemma holds.
Second, if the comparison between $u'$ and $v$
results in Case~2, then $v^*$ cannot
be in  $v.\ell$, if $v$ is an inner node.
Also, if $u \in \pi_\ell$, then necessarily also
$u.\ell \in \pi_\ell$, since Case~1
occurred between $u$ and $v$. If $v \in \pi_r$,  
since Case~2 occurred between $u'$ and $v$, the node
$v$ cannot
be a leaf and $v.r \in \pi_r$. Thus, in either case
the invariant is maintained and we discover a new
node on $\pi_\ell$ or on $\pi_r$.
Third, assume the comparison between
$u'$ and $v$ results in Case~3. If
$u \in \pi_\ell$, then also $u.\ell \in \pi_\ell$,
because $u.r \in \pi_\ell$ was excluded by
the comparison between $u$ and $v$. In this case,
the lemma holds. If $u \not\in \pi_\ell$,
then also $u'.r \not \in \pi_\ell$, so the fact
that Case~3 occurred between $u'$ and $v$ implies that
$v.\ell$ must be on $\pi_r$ (in this case,
$v$ cannot be a leaf, since otherwise we would
have $v^* = v$ and Lemma~\ref{lem:leaf_inv} would
give $u'.r \in \pi_\ell$, which we have already ruled out).
The argument for Case~2 is symmetric.
\end{proof}

\begin{figure}
    \centering
	\begin{subfigure}{0.801\textwidth}
        \includegraphics[width=\textwidth]{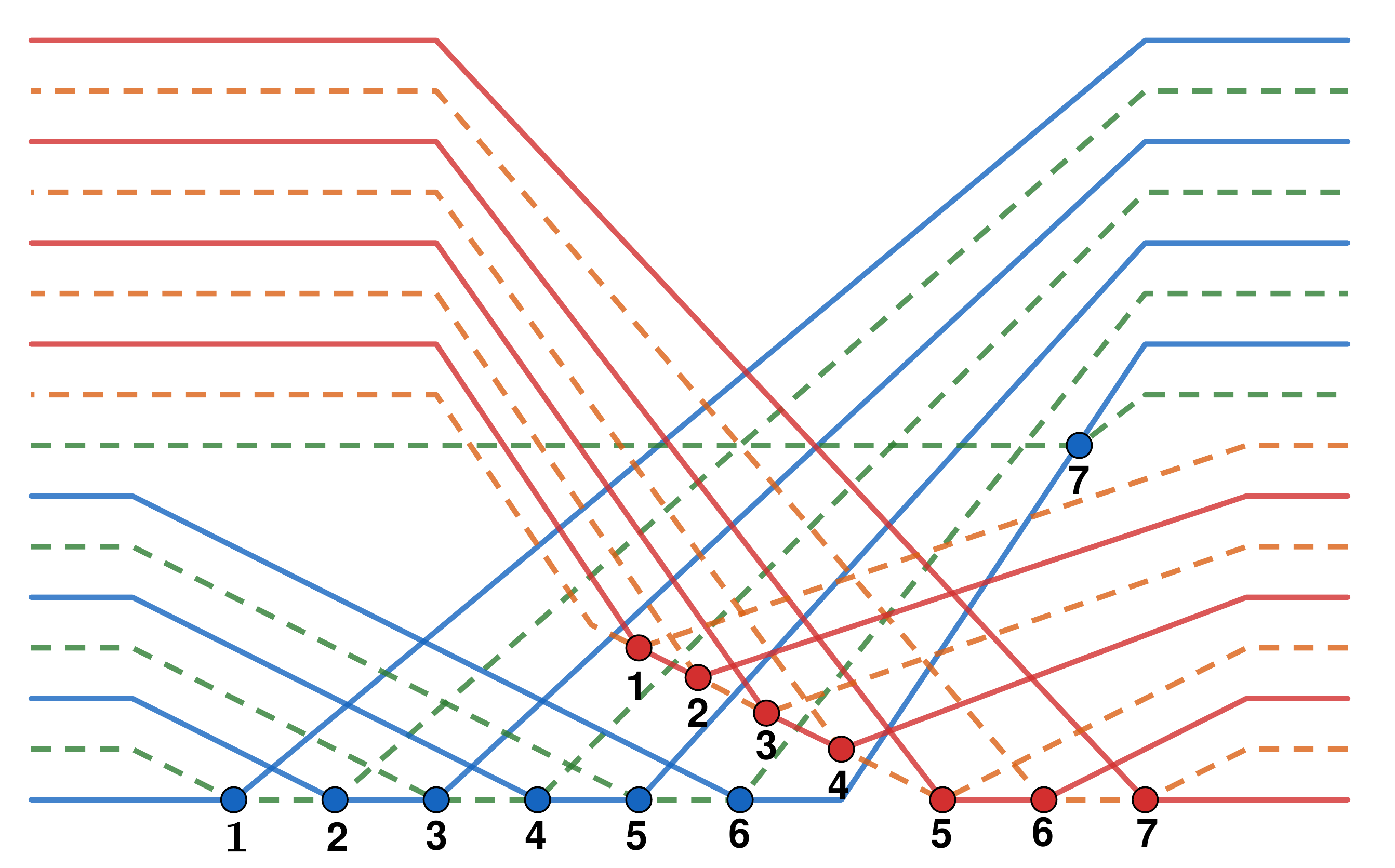}
        \caption{Demonstration of two set of pseudo-lines and their lower envelope: (i) the blue and green pseudo-lines, (ii) the red and orange pseudo-lines. The blue and the red dots represents the intersection points on the lower envelopes.}
    	\label{f:intersection_point_demo2} 
	\end{subfigure}
	\begin{subfigure}{0.8\textwidth}
        \includegraphics[width=\textwidth]{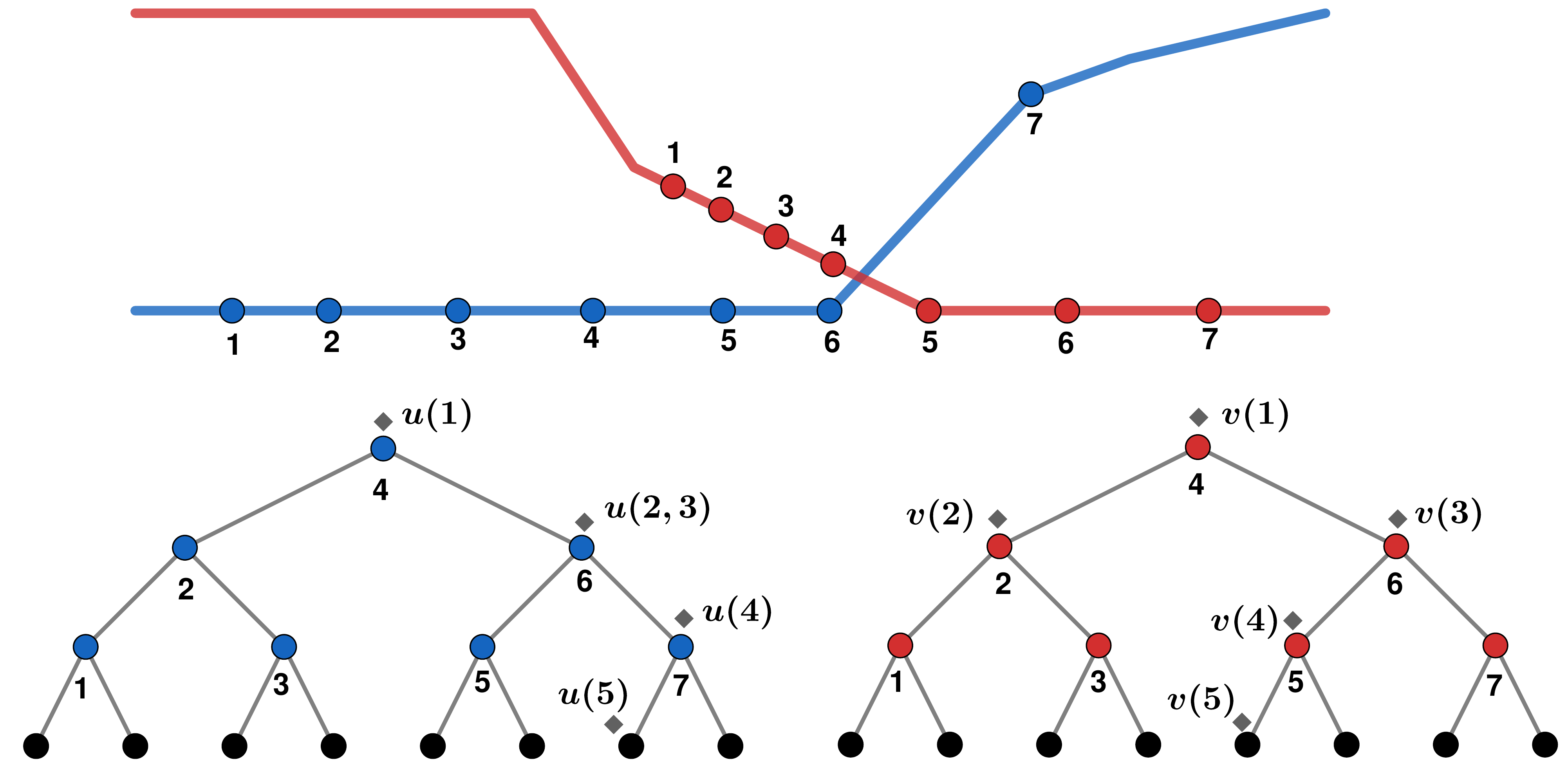}
        \caption{The top figure shows the lower envelope of (a). The bottom figure shows the the trees which maintain the lower envelopes. $u(i)$ and $v(i)$ shows the position of the pointers $u$ and $v$ at step $i$, during the search procedure.}
    	\label{f:intersection_point_demo} 
	\end{subfigure}
    \caption{Example of finding the intersection point of two lower envelopes:}
    \begin{tabular}{|c|cc|cc|c|}
    \hline 
        Step  & $u$ & $v$ & uStack & vStack & Procedure case \\\hline
        1  & 4 & 4 & $\emptyset$ & $\emptyset$ & Case~3 \\
        2  & 6 & 2 & 4 & 4 & Case~2 $\rightarrow$ Case~2 \\
        3  & 6 & 6 & 4 & $\emptyset$ & Case~3 \\
        4  & 7 & 5 & 4, 6 & 6 & Case~1 $\rightarrow$ Case~3 \\
        5  & 7* & 5* & 4, 6 & 6, 5 & Case~3 $\rightarrow$ End\\\hline
    \end{tabular}
    \label{fig:algodemo}
\end{figure}

\begin{lemma}
    The intersection point $q$ between $\mathcal{L}_\ell$ and 
    $\mathcal{L}_r$ can be found in $O(\log n)$ time.
\end{lemma}

\begin{proof}
    In each step, we either discover a new node of 
    $\pi_\ell$ or of $\pi_r$, or we pop an element
    from \texttt{uStack} or \texttt{vStack}. 
    Elements are pushed only when
    at least one new  node on $\pi_\ell$ or 
    $\pi_r$ is discovered. 
    As $\pi_\ell$ and $\pi_r$ are each a path from the root to a leaf in a balanced binary tree, 
    we  need $O(\log n)$
    steps.
\end{proof}

\subparagraph*{Acknowledgments.} We thank Haim Kaplan and Micha Sharir for helpful discussions. Work by P.A. has been supported by
NSF under grants CCF-15-13816, CCF-15-46392, and IIS-14-08846, by ARO 
grant W911NF-15-1-0408, and by grant 2012/229 from the U.S.-Israel 
Binational Science Foundation. Work by D.H.\ and R.C.\ has been supported in part by the Israel Science Foundation
(grant no.~825/15), by the Blavatnik Computer Science Research Fund,
by the Blavatnik Interdisciplinary Cyber Research Center at Tel Aviv
University, and by grants from Yandex and from Facebook. Work by W.M. has been partially supported by ERC STG 757609 and GIF grant 1367/2016.

\bibliography{eurocg19_example}

\appendix

\newpage

\end{document}